\documentclass[a4paper,UKenglish,cleveref, autoref]{lipics-v2021}
\nolinenumbers
\usepackage{todonotes}
\usepackage{xspace}
\usepackage[export]{adjustbox}

\title{Bounding the Treewidth of Outer $k$-Planar Graphs via Triangulations}

\author{Oksana Firman}{Universit\"at W\"urzburg, Germany \and
  \url{https://www.informatik.uni-wuerzburg.de/en/algo/team/firman-oksana}}{}{https://orcid.org/0000-0002-9450-7640}{Supported
  by DFG grant Wo758/9-1.}

\author{Grzegorz Gutowski}{Institute of Theoretical Computer Science, Faculty of Mathematics and Computer Science, Jagiellonian University, Krak{\'o}w, Poland}{grzegorz.gutowski@uj.edu.pl}{https://orcid.org/0000-0003-3313-1237}{The research cooperation was funded by the program Excellence Initiative -- Research University at the Jagiellonian University in Kraków.}

\author{Myroslav~Kryven}{University of Manitoba, Canada}{myroslav.kryven@umanitoba.ca}{https://orcid.org/0000-0003-4778-3703}{}

\author{Yuto Okada}{Nagoya University, Japan \and \url{https://yutookada.com/en}}{okada.yuto.b3@s.mail.nagoya-u.ac.jp}{https://orcid.org/0000-0002-1156-0383}{Supported by JST SPRING, Grant Number JPMJSP2125 and JSPS KAKENHI, Grant Number JP22H00513 (Hirotaka Ono).}

\author{Alexander~Wolff}{Universit\"at W\"urzburg, Germany \and \url{https://www.informatik.uni-wuerzburg.de/en/algo/team/wolff-alexander}}{}{http://orcid.org/0000-0001-5872-718X}{}

\authorrunning{O.~Firman, G.~Gutowski, M.~Kryven, Y.~Okada, and A.~Wolff}

\Copyright{Oksana Firman, Grzegorz Gutowski, Myroslav Kryven, Yuto Okada, and Alexander Wolff}



\hideLIPIcs

\acknowledgements{We thank Yota Otachi for suggestions that helped us to improve the lower bound. We also thank Hirotaka Ono for supporting our work.}

\crefname{observation}{Observation}{Observations}

\DeclareMathOperator{\lcr}{lcr}
\newcommand{\lcro}{\ensuremath{\lcr^\circ}}
\DeclareMathOperator{\tw}{tw}

\DeclareMathOperator{\cop}{cop}
\DeclareMathOperator{\sn}{sn}

\newcommand{\T}{\ensuremath{\mathcal{T}}\xspace}

\definecolor{defblue}{rgb}{0.121,0.47,0.705}
\definecolor{linkblue}{rgb}{0.098,0.098,0.4392}
\let\emph\relax
\DeclareTextFontCommand{\emph}{\color{defblue}\em}

\graphicspath{{./figures/}}

\begin{document}

\maketitle

\begin{abstract}
  The \emph{treewidth} is a structural parameter that measures the
  tree-likeness of a graph.  Many algorithmic and combinatorial
  results are expressed in terms of the treewidth.  In this paper, we
  study the treewidth of \emph{outer $k$-planar} graphs, that is, 
  graphs that admit a straight-line drawing where all the vertices
  lie on a circle, and every edge is crossed by at most $k$ other edges.
  
  Wood and Telle [New York J.\ Math., 2007] showed that every outer
  $k$-planar graph has treewidth at most $3k + 11$ using so-called
  planar decompositions, and later, Auer et al.\ [Algorithmica, 2016]
  proved that the treewidth of outer $1$-planar graphs is at most $3$, which is tight.

  In this paper, we improve the general upper bound to $1.5k + 2$ and give a tight bound of $4$ for $k = 2$.
  We also establish a lower bound: we show that, for
  every even $k$, there is an outer $k$-planar graph with
  treewidth~$k+2$.  Our new bound immediately implies a better bound on
  the \emph{cop number}, which answers an open question of Durocher et al.~[GD 2023] in the affirmative.

  Our treewidth bound relies on a new and simple
  triangulation method for outer $k$-planar graphs that yields few
  crossings with graph edges per edge of the triangulation.  Our method also enables us to
  obtain a tight upper bound of $k + 2$ for the \emph{separation number} of
  outer $k$-planar graphs, improving an upper bound of $2k + 3$ by
  Chaplick et al.\ [GD 2017].
  We also consider \emph{outer min-$k$-planar} graphs, a
  generalization of outer $k$-planar graphs, where we
  achieve smaller improvements.
\end{abstract}

\section{Introduction}

\emph{Treewidth} measures the tree-likeness of a graph via so-called
tree decompositions.
A \emph{tree decomposition} of a graph~$G$ covers the vertex set of~$G$ by bags such that every edge is in some bag and the bags form a tree such that, for every vertex~$v$ of~$G$, the bags that contain~$v$ form a subtree.
The width of a tree decomposition is the maximum size of a bag minus one.
The treewidth of~$G$ is the minimum width over all tree
decompositions of~$G$.

Treewidth is an important structural parameter because
the running time of many graph algorithms and algorithms for drawing
graphs depend on the treewidth.  Recently treewidth-based techniques
have been successfully applied for \emph{convex} drawings, that is,
straight-line drawings where vertices lie on a circle.  For example,
Bannister and Eppstein~\cite{be-cm1p2pdgbtw-JGAA18} (\cite{be-cm1p2p-GD14})
showed that if a graph admits a convex drawing with at most $k$
crossings in total, then its treewidth is bounded by a function
of~$k$.  Furthermore, they used this fact to give, for a fixed $k$, a
linear-time algorithm (via extended monadic second-order logic and
Courcelle’s theorem~\cite{courcelle1990}) to decide whether a convex
drawing with at most $k$ crossings exists for a given graph.
Chaplick et al.~\cite{bundles} (\cite{cdkprw-bcr-GD19})
generalized this to \emph{bundled crossings}
(a crossing of two bundles of edges such that in each
bundle the edges travel in parallel, which is counted as one
crossing).  Another prominent class of graphs
with convex drawings is the class of
\emph{outer $k$-planar} graphs, that is, graphs that admit a convex
drawing where each edge is crossed at most $k$ times.  Unlike
\emph{$k$-planar} graphs (without the restriction on the placement of
vertices), the treewidth of outer $k$-planar graphs can be bounded by
a function of~$k$ only (see discussion below).
Similarly to Bannister and Eppstein~\cite{be-cm1p2pdgbtw-JGAA18},
Chaplick et al.~\cite{BeyondOuterplanarity} used this %
to test in linear time, for any fixed~$k$, whether a given graph is
\emph{full outer $k$-planar}, i.e., whether it admits an outer
$k$-planar drawing where no crossing appears on the boundary of the
outer face.

For disambiguation, recall that \emph{$k$-outerplanar} graphs are
defined as follows.  A drawing of a graph is $1$-outerplanar if
all vertices of the graph lie on the outer face.  For $k>1$, a drawing
is $k$-outerplanar if deleting the vertices on the outer face
yields a $(k-1)$-outerplanar drawing.  For $k \ge 1$, a graph is
$k$-outerplanar if it admits a $k$-outerplanar drawing.

In this paper, we are particularly interested in the treewidth of
outer $k$-planar graphs.
Note that every graph is outer $k$-planar for some value of~$k$.
For a graph $G$, let $\lcro(G)$ denote the \emph{convex local crossing
  number} of~$G$~-- the smallest $k$ such that $G$ is outer $k$-planar. 
Consult Schaefer's survey~\cite{Schaefer2024} for details on
this and many other types of crossing numbers.
Outer $k$-planar graphs admit balanced separators of size $O(k)$ (more
below) associated with the drawing, which makes it possible to test, for any fixed $k$, outer $k$-planarity in quasi-polynomial time \cite{BeyondOuterplanarity}.
This implies that the recognition problem is not NP-hard unless the
Exponential Time Hypothesis fails~\cite{BeyondOuterplanarity}.

We also consider a generalization of outer $k$-planar graphs, namely
\emph{outer min-$k$-planar graphs}, which are graphs that admit a
convex drawing where, for every pair of crossing edges, at least one
edge is crossed at most $k$ times.  Wood and Telle
\cite[Proposition~8.5]{wt-pdcng-NYJM07} showed that any (min-) outer
$k$-planar graph has treewidth at most $3k+11$.
More precisely, they
showed that, for every outer min-$k$-planar graph~$G$, one obtains a
tree decomposition of width at most $3k+11$ from a \emph{planar
  decomposition} of~$G$, which is a generalization of a tree decomposition.
The proof uses a result by Bodlaender~\cite{Bodlaender1998} on $k$-outerplanar graphs.

For constant~$k$, better treewidth bounds are known.  A folklore
result is that outerplanar graphs ($k=0$) have treewidth at most~$2$.
Auer et al.~\cite{auer2016outer,abbghnr-co1pg-Algorithmica21} (\cite{abbghnr-ro1pg-GD13})
showed that maximal outer $1$-planar graphs are chordal graphs and,
therefore, have treewidth at most~$3$, which is tight.

Another parameter linked to the treewidth is the separation number of
a graph.  A pair of vertex sets $(A, B)$ is a \emph{separation} of a
graph~$G$ if $A \cup B = V(G)$ and there is no edge between
$A \setminus B$ and $B \setminus A$.  A separation $(A, B)$ is said to
be \emph{balanced} if the sizes of both $A \setminus B$ and
$B \setminus A$ are at most $2n / 3$, where $n$ is $|V(G)|$.  For a (balanced) separation
$(A, B)$, the set $A \cap B$ is called (balanced) \emph{separator},
and $|A \cap B|$ is the \emph{order} of~$(A, B)$.  The
\emph{separation number} of $G$, $\sn(G)$, is the minimum integer $k$ such that every subgraph of $G$ has a balanced separation of order $k$.
Note that, for any subgraph $G'$ of $G$, $\tw(G') \le \tw(G)$.
This implies that, for every graph~$G$,
$\sn(G) \le \tw(G) + 1$~\cite{dn-tgbs-JCTB19}.
On the other hand, Dvo\v{r}\'ak and Norin~\cite{dn-tgbs-JCTB19}
showed that, for every graph~$G$, $\tw(G) \le 15\sn(G)$.

\subparagraph*{Our contribution.}

Given an outer $k$-planar drawing~$\Gamma$ of a graph~$G$, let the
\emph{outer cycle} of~$G$ be the cycle that connects the vertices
of~$G$ in the order along the circle on which they lie in~$\Gamma$
(even if $G$ does not contain all edges of this cycle).
We introduce two simple methods to construct, given an outer $k$-planar
drawing~$\Gamma$ of a graph~$G$, a triangulation of the outer cycle
of~$G$ with the property that each edge of the triangulation is
crossed by at most $k$ edges of~$G$; see \cref{sec:triangulations}.
The resulting triangulations yield the following bounds; see
\cref{sec:applications}.
\begin{itemize}
\item We improve the upper bound of Wood and Telle
  \cite{wt-pdcng-NYJM07} regarding the treewidth of outer $k$-planar
  graphs from $3k+11$ to $1.5k + 2$ (\cref{thm:okp-tw1.5k}
        in \cref{sec:treewidth}).  Our proof is constructive and implies a practical algorithm; the tree
  decomposition that we obtain follows the weak dual of the triangulation.
  For outer $2$-planar graphs, our methods yield an upper bound of~4
  (\cref{thm:o2p-tw4}), which is tight due to~$K_5$.
\item Chaplick et al.~\cite{BeyondOuterplanarity} showed that,
  for every outer $k$-planar graph~$G$, %
  its separation number is at most $2k+3$.
  We improve this upper bound to $k + 2$; see \cref{sec:separation}.
\item We give new lower bounds of $k + 2$ for both the treewidth and
  the separation number of outer $k$-planar graphs; see
  \cref{sec:lower}.  Note that the latter bound is tight (\cref{thm:sn-lb-okp}).
\item Durocher et al.~\cite{durocher2023cops} recently proved that the
  cop number of general 1-planar graphs is not bounded, but the
  maximal 1-planar graphs have cop number at most~3.  Since it is
  known~\cite{jkt-crggf-CDM10} that for every graph~$G$,
  $\cop(G) \le \tw(G)/2+1$, Durocher et al.~\cite{durocher2023cops}
  observed that the treewidth bound of Wood and Telle yields, for
  every outer $k$-planar graph~$G$, that $\cop(G) \le 1.5k + 6.5$.
  They asked explicitly whether the multiplicative factor of $1.5$ can
  be improved.  We answer this question in the affirmative since our
  new treewidth bound immediately yields the better bound
  $\cop(G) \le 0.75k+1.75$.
\item For outer min-$k$-planar graphs, our triangulation improves the
  bound of Wood and Telle $3k+11$ slightly to $3k+1$ and gives a bound
  $2k + 1$ on the separation number; see \cref{sec:applications}.
\end{itemize}

\subparagraph*{Related results.}
A structurally similar type of result is known for \emph{$2$-layer
  $k$-planar} graphs.  A $2$-layer $k$-planar graph is a bipartite
graph that admits a $k$-planar straight-line drawing with
vertices placed on two parallel lines.
The maximal $2$-layer $0$-planar graphs (called \emph{caterpillars}) are
the maximal pathwidth-1 graphs~\cite{kinnersley1994,proskurowski1999}.
Angelini et al.~\cite{2LayerkPlanar} (\cite{alfs-2lkpg-GD20})
showed that $2$-layer $k$-planar graphs have pathwidth at most $k + 1$
and gave a lower bound of $(k + 3)/2$ for every odd number~$k$.
Hence, their result can be considered as a smooth
generalization of the caterpillar result.
Similarly, our treewidth bound $1.5k + 2$ for outer $k$-planar graphs
smoothly generalizes the treewidth-$2$ bound for maximal outerplanar graphs.

\section{Preliminaries}
\label{sec:preliminaries}

A \emph{convex drawing} $\Gamma$ of a graph is a straight-line drawing
where the vertices of the graph are placed on different points of a
circle, %
which we call the \emph{circle of~$\Gamma$}.
An \emph{outer $k$-planar drawing}~$\Gamma$ of a graph is a convex
drawing such that every edge crosses at most $k$ other edges.
A counterclockwise walk of the circle of $\Gamma$ yields a cyclic
order on the vertices of the graph.
We say that two sets on distinct four vertices $\{v_1,w_1\}$, and
$\{v_2,w_2\}$ are \emph{intertwined} in~$\Gamma$ if they are ordered
$\langle v_1,v_2,w_1,w_2 \rangle$ or $\langle v_1,w_2,v_2,w_1 \rangle$
in this cyclic order.
Observe that two edges $\{v_1,w_1\}$ and $\{v_2,w_2\}$ cross
in~$\Gamma$ if and only if they span four different vertices and are
intertwined.
In the remainder of this section, we make some simple observations
that will be helpful later; they can be skipped by more experienced
readers.
\begin{observation}\label{obs:moving-points-on-circle}
  Let $G$ be a graph.  If $\Gamma$ is an outer $k$-planar drawing
  of~$G$, then every convex drawing~$\Gamma'$ of $G$ with the same
  cyclic vertex order as in $\Gamma$ is also an outer $k$-planar
  drawing of~$G$.
\end{observation}
\begin{proof}
  As the cyclic orders on vertices of~$G$ defined by~$\Gamma$ and
  $\Gamma'$ are the same, exactly the same pairs of edges cross
  in~$\Gamma$ and in~$\Gamma'$.
\end{proof}
In the definition of convex drawings, we could have allowed placing vertices on
arbitrary convex shapes (instead of a circle) and drawing edges as
curves (instead of straight-line segments)~--
as long as every curve is drawn inside the circle and no two curves
cross more than once.  Using curves will be sometimes convenient in
our proofs later.

\begin{observation}\label{obs:straight-edges}
  Let $G$ be a graph, let $\Gamma$ be an outer $k$-planar drawing
  of~$G$, and let $c$ be a curve that joins vertices $v$ and $w$
  of~$G$ and is contained inside the circle of~$\Gamma$.  If $c$
  crosses at most $l$ edges of~$\Gamma$, then the straight-line
  segment $\overline{vw}$ also crosses at most $l$ edges of~$\Gamma$.
\end{observation}

\begin{proof}
  As $c$ must cross every edge $\{v',w'\}$ of~$G$ that is intertwined
  with $\{v,w\}$, we get that $c$ must have at least one crossing with
  every edge crossed by the straight-line segment~$\overline{vw}$.
\end{proof}

A graph~$G$ is \emph{maximal outer $k$-planar} if $G$ is outer
$k$-planar and $G$ does not contain any vertex pair $e \in V^2(G)
\setminus E(G)$ such that the graph $(V(G), E(G) \cup \{e\})$ is still
outer $k$-planar.
Since removing edges increases neither treewidth nor separation number,
we are interested in properties of maximal outer $k$-planar graphs.

\begin{observation}\label{obs:outer-cycle}
  If $G$ is a maximal outer $k$-planar graph with at least three
  vertices, then, in every outer $k$-planar drawing of~$G$, the outer
  face is bounded by a simple cycle.
\end{observation}

\begin{proof}
  Let $v$ and $w$ be two vertices that are consecutive in the cyclic
  order defined by some outer $k$-planar drawing~$\Gamma$ of~$G$.  For
  a contradiction, suppose that $v$ and $w$ are not adjacent.  As $v$
  and $w$ are consecutive on the circle of~$\Gamma$, the set $\{v,w\}$
  is not intertwined with any edge in~$\Gamma$.  Thus, $\Gamma$ can be
  extended to an outer $k$-planar drawing of $G+\{v,w\}$, which
  contradicts the maximality of~$G$.
\end{proof}

We remark that
\cref{obs:moving-points-on-circle,obs:straight-edges,obs:outer-cycle}
analogously hold for outer min-$k$-planar graphs.

\section{Triangulations}
\label{sec:triangulations}

In this section, we present two simple strategies to triangulate maximal outer $k$-planar graphs.
These triangulations will serve as tools to construct tree decompositions and balanced separators in \cref{sec:applications}.
We assume that some outer $k$-planar drawing~$\Gamma$ of a graph $G$ on $n$ vertices is given.
During the triangulation procedures, we will modify the drawing, but we will never change the cyclic order of the vertices.
Thus, let $v_1,v_2,\dots,v_n$ be the vertices of~$G$ in the cyclic
(counterclockwise) order defined by~$\Gamma$.
As we focus on maximal graphs, by \cref{obs:outer-cycle}, we have that
$\langle v_1,v_2,\dots,v_n,v_1 \rangle$ is a cycle in $G$.
Our goal is to construct a triangulation~\T of this cycle such that
the edges of~\T cross only a limited number of edges of~$G$.
We refer to the edges of~\T as \emph{links} in order
to distinguish them from the edges of~$G$.  Every edge of the outer
cycle is a link in~\T; we call these links \emph{outer} links.
We will select a set of $n-3$ other pairwise non-intertwined pairs of
vertices as \emph{inner} links that, together with the outer links,
form a triangulated $n$-gon.
Note that some of the inner links may coincide with edges of~$G$.
We say that a link is \emph{pierced} by an edge of~$G$ if their
endpoints are intertwined in the cyclic order.
The \emph{piercing number} of a link is the number of edges of~$G$
that pierce the link.
In particular, the piercing number of
outer links is~0, and if a link coincides with an edge, then its
piercing number is at most $k$ by the outer $k$-planarity
of~$\Gamma$.  We define the \emph{edge piercing number} of~\T to be
the maximum piercing number of any link of~\T.
In \cref{sec:applications} we show how to use triangulations with
small edge piercing number.

Our \emph{splitting procedure} starts by declaring the
link~$\{v_1,v_n\}$ active and considers all other vertices to lie on
the \emph{right} side of the link. See \cref{fig:algo-steps} for illustration.
Clearly, the link~$\{v_1,v_n\}$ is
not pierced.  In each recursive step, the input is an active link
$\{v_i,v_k\}$ with $1 \le i < k \le n$ (promised to be pierced at most
some limited number of times), along with a distinguished right side
$v_{i+1},\dots,v_{k-1}$, where currently no inner links have been
selected apart from $\{v_i,v_k\}$.  (Vertices
$v_1,\dots,v_{i-1},v_{k+1},\dots,v_n$ form the \emph{left} side.)  The
goal is to pick, among the vertices on the right side, a \emph{split
  vertex}~$v_j$ such that the new links $\{v_i,v_j\}$ and
$\{v_j,v_k\}$ are pierced by at most some limited number of edges.  The new links are
added to~\T (completing the triangle $v_i v_j v_k$) and become active.
The link $\{v_i,v_k\}$ ceases to be active.  The split gives rise to
two new splitting instances $\{v_i,v_j\}$ and $\{v_j,v_k\}$, which are
then solved recursively.

\begin{figure}
        \centering
	\begin{subfigure}[b]{.195\textwidth}
		\centering
		\includegraphics[page=1]{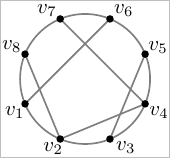}
		\subcaption{input graph}
		\label{fig:algo-steps-graph}
	\end{subfigure}
        \hspace{.5mm}
	\begin{subfigure}[b]{.17\textwidth}
		\centering
		\includegraphics[page=2]{algo-steps}
		\subcaption{\nolinenumbers{}step 1}
		\label{fig:algo-step1}
	\end{subfigure}
        \hspace{.5mm}
	\begin{subfigure}[b]{.17\textwidth}
		\centering
		\includegraphics[page=3]{algo-steps}
		\subcaption{\nolinenumbers{}step 2}
		\label{fig:algo-step2}
	\end{subfigure}
        \hspace{.5mm}
	\begin{subfigure}[b]{.17\textwidth}
		\centering
		\includegraphics[page=4]{algo-steps}
		\subcaption{\nolinenumbers{}step 3}
		\label{fig:algo-step3}
	\end{subfigure}
        \hspace{.5mm}
	\begin{subfigure}[b]{.195\textwidth}
		\centering
		\includegraphics[page=5]{algo-steps}
		\subcaption{\nolinenumbers{}triangulation}
		\label{fig:algo-steps-triangulation}
	\end{subfigure}

	\caption{Steps of the splitting procedure.  The active edge is
          purple, the (new/old) links are (dark/light) blue, the
          left/right side of the active edge is red/green.  The
          split vertex is big.}
	\label{fig:algo-steps}
\end{figure}

The set of edges of~$G$ that pierce the active link are the
\emph{piercing edges}.
Note that every piercing edge has one left and one right endpoint.
Intuitively, the edges with two left endpoints are of no concern to
the splitting procedure.  The next lemma allows us to push all
crossings among piercing edges to the left of $\{v_i,v_k\}$.

\begin{lemma}\label{lem:push-the-crossings}
    Given an outer $k$-planar drawing~$\Gamma$ of a graph $G$ and an active link $\{v_i,v_k\}$, there exists an outer $k$-planar drawing $\Gamma'$ with the same cyclic order as~$\Gamma$, and all crossings among the edges piercing the active link are drawn to the left of the active link.
\end{lemma}

\begin{proof}
  We draw the vertices of $G$ on the unit circle, keeping the cyclic
  order defined by~$\Gamma$.  We set $\varepsilon>0$ to a sufficiently
  small value.  Then we place $v_i$ at $(0,-1)$, $v_k$ at $(0,1)$, all
  the vertices on the left side of the link within arc
  distance~$\varepsilon$ of $(-1, 0)$ and
  spread the vertices from the right side of the link evenly along the right semicircle, see
  \cref{fig:pushing-crossings}.
    
  \begin{figure}
    \centering
    \includegraphics{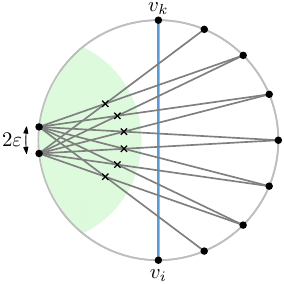}
    \caption{We can push crossings among the edges that pierce the
      active link $\{v_i,v_k\}$ to the left.  We show only the top- and bottommost endpoints on the left.  Crossings
      occur exclusively in the green area, which can be made arbitrarily
      small.}
    \label{fig:pushing-crossings}
  \end{figure}
   
  Now, \cref{obs:moving-points-on-circle} yields that the new drawing
  is an outer $k$-planar drawing of~$G$ with the same cyclic order.
  As every pair of intertwined piercing edges has their left
  endpoints in arc distance at most $\varepsilon$ from $(-1,0)$ and
  their right endpoints in arc distance at least $\pi / n$ from each
  other, the crossing of two such segments is to the left of the
  active link.
\end{proof}

\cref{lem:push-the-crossings} gives us an equivalent drawing of~$G$ in
which piercing edges cross the active link in the same order from
bottom to top as their right endpoints occur in the cyclic order.  This
allows us to draw the piercing edges without crossings on the right
side of the active link.

Now, we are ready to present our first and basic triangulation method.
Asymptotically, it yields a worse bound than
\cref{lem:triangulation-strong}, but it serves as a gentle
introduction into our techniques, and it will be used for small values
of $k$ and for outer min-$k$-planar graphs later.

\begin{lemma}
  \label{lem:triangulation-weak}
  For $k \ge 1$, every outer $k$-planar drawing of a maximal outer
  $k$-planar graph admits a triangulation of the outer cycle with edge
  piercing number at most $2k-1$.
\end{lemma}

\begin{proof}
    We prove the lemma by providing an appropriate splitting procedure
    that recursively constructs the triangulation.
    Let~$\Gamma$ be an outer $k$-planar drawing of the given maximal
    outer $k$-planar graph~$G$.
    We prove by induction on the number of steps of the splitting
    procedure that every active link has piercing number at most $2k-1$.
    The claim is true in the beginning, as we start with the link
    $\{v_1,v_n\}$, which is not pierced at all.

    Now we consider a later step of the procedure.  Let
    $\lambda=\{v_i,v_k\}$ denote the active link.  For brevity, set
    $u=v_i$ and $v=v_k$.  There are two cases.
    
    In the first case, $\lambda$ is not pierced by any edge of~$G$.
    Observe that $\{v_i,v_{i+1}\}$ is an edge, so~$u$ has at least one
    neighbor on the right side of $\lambda$ in~$G$.  We select the
    split vertex, $w$, to be the neighbor of~$u$ from the right side
    of $\lambda$ that is closest to~$v$ in the circular
    order; see \cref{fig:okp-proof-1}.  In this case, the link
    $\{u,w\}$ coincides with an edge of~$G$, and thus is pierced by at
    most $k$ edges.  To see that the link $\{w,v\}$ is also pierced by
    at most $k$ edges, let $c''$ be a curve that starts in $w$, goes
    along $\{w,u\}$ towards $u$ on the right side (in the direction of
    walking), then follows $\{u,v\}$ on the right side, and ends
    at~$v$.  This curve crosses only the edges that pierce the link
    $\{u,w\}$, so at most $k$ edges pierce the link $\{w,v\}$ by
    \cref{obs:straight-edges}.

    In the second case, $\lambda$ is pierced by edges of~$G$.  Using
    \cref{lem:push-the-crossings}, we modify the drawing~$\Gamma$ so
    that the piercing edges cross~$\lambda$ in the same bottom-to-top
    order as their right endpoints.  Let $e_1, e_2, \dots, e_\ell$ denote
    the piercing edges, and let $w_1, w_2, \dots, w_\ell$ denote their
    respective, not necessarily different, right endpoints in the
    order from bottom to top.  By induction, we have %
    $\ell \le 2k-1$.  Then we pick the middle edge $e_{j}$, where
    $j = \lceil (\ell+1)/2 \rceil$, and set the split vertex $w = w_j$.

    \begin{figure}
    \begin{subfigure}[b]{.37\textwidth}
        \centering
        \includegraphics[page=1]{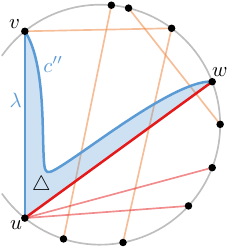}
        \subcaption{constructing the new 
                  triangle~$\triangle$ if $\lambda$ is not pierced by any other edge of $G$.}
        \label{fig:okp-proof-1}
    \end{subfigure}
    \hfill
    \begin{subfigure}[b]{.6\textwidth}
        \centering
        \includegraphics[page=2]{okp-proof}
        \subcaption{\nolinenumbers{}constructing the new 
                  triangle~$\triangle$
                  bounded by~$\lambda$, $c'$, and~$c''$;
                  the orange edges that pierce both~$c'$
                  and~$c''$ also cross the edge~$e_j$.}
        \label{fig:okp-proof-2}
    \end{subfigure}
    \caption{Constructing a triangulation~\T with
          edge piercing number at most $2k-1$.}
    \label{fig:okp-proof}
    \end{figure}  

    To bound the number of edges that pierce the link $\{u, w\}$, we draw the curve $c'$ as in \cref{fig:okp-proof-2}.
    We start the curve $c'$ at $u$ and follow $\lambda$ until before it crosses $e_j$, then follow $e_j$ until we arrive at $w$.
    The curve $c'$ crosses at most $\lceil (\ell-1)/2 \rceil$ edges
    that pierce $\lambda$ and at most~$k$ edges that pierce~$e_j$.
    Thus, there are at most $\lceil (\ell-1)/2 \rceil + k \leq \lceil (2k-2)/2 \rceil + k = 2k-1$ edges that pierce the link $\{u,w\}$. 
    We can argue symmetrically for the link $\{w,v\}$; see curve~$c''$ in \cref{fig:okp-proof-2}.

    As we have bounded the piercing number of every active link during
    the triangulation procedure by $2k-1$, we get the desired bound of $2k-1$
    for the edge piercing number of the constructed triangulation.
\end{proof}

Next, we present our refined triangulation method, where we consider
an additional edge to get a better bound for the edge piercing number.

\begin{lemma}
  \label{lem:triangulation-strong}
  For every $k \ge 1$, every outer $k$-planar drawing of a maximal
  outer $k$-planar graph admits a triangulation of the outer cycle
  with edge piercing number at most $k$.
\end{lemma}

\begin{proof}
    The proof is similar to the proof of
    \cref{lem:triangulation-weak}, and we use the same terminology.
    Let again~$\Gamma$ be an outer $k$-planar drawing of the given
    maximal outer $k$-planar graph~$G$.
    We present a splitting procedure and show that every active link
    is pierced by at most~$k$ edges.  For the base of the induction,
    observe that the link $\{v_1,v_n\}$ has no piercing edges.
    Now, let $\lambda = \{v_i,v_k\}$ be the active link and set $u=v_i$ and $v=v_k$.
    If $\lambda$ has no piercing edges in~$\Gamma$, we proceed as in
    the proof of \cref{lem:triangulation-weak} and pick the
    split vertex to be the neighbor of $u$ from the right
    side of $\lambda$ that is closest to $v$.
    Otherwise, $\lambda$ is pierced by edges of $G$.
    Using \cref{lem:push-the-crossings}, we assume that the piercing edges cross the active link in the same bottom-to-top order as their right endpoints.
    Now, let $e_1,e_2,\ldots,e_\ell$ be the piercing edges, and let
    $w_1, w_2, \dots, w_\ell$ be their respective, not necessarily
    different, right endpoints in the order from bottom to top.
    We pick the middle edge~$e_j$ with $j=\lceil(\ell+1)/2\rceil$.
    If $e_j$ has no crossings on the right side of the active link,
    then we pick~$w_j$ to be the split vertex.
    The two curves that start in~$w_j$, follow~$e_j$ to~$e$,
    and then follow~$\lambda$ to either $u$ or $v$ have
    at most $\lceil(\ell-1)/2\rceil \le k$ crossings each, so we can
    proceed.
    Now, let $x$ denote the first crossing on the edge $e_j$ that occurs
    on the right side of $\lambda$, and let $\hat e$ denote the edge
    that crosses~$e_j$ at $x$; see \cref{fig:okp-1.5k-proof-1}.

    \begin{figure}
        \begin{subfigure}[b]{.41\textwidth}
          \centering
          \includegraphics[page=1]{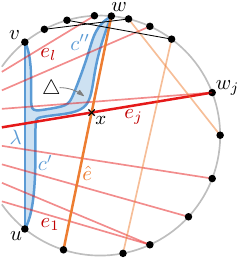}
          \subcaption{$\lambda$ and $\hat e$ do not share an
            endpoint; the new triangle~$\triangle$ is bounded
            by~$\lambda$, $c'$, and~$c''$.}
          \label{fig:okp-1.5k-proof-1}
        \end{subfigure}
        \hfill
        \begin{subfigure}[b]{.4\textwidth}
          \centering
          \includegraphics[page=2]{okp-1.5k-proof}
          \subcaption{$\lambda$ and $\hat e$ share the endpoint $u$;
            the new triangle~$\triangle$ is bounded by~$\lambda$,
            $\hat e$, and~$c''$.}
          \label{fig:okp-1.5k-proof-2}
        \end{subfigure}
        \caption{Constructing a triangulation \T with edge
          piercing number $k$.}
        \label{fig:okp-1.5k-proof}
    \end{figure}

    We first consider the case where $\hat e$ and $\lambda$ are disjoint.
    Let $a$ and $b$ be the numbers
    of crossings on $\hat e$ on the two sides of $x$.
    As $a+b+1 \le k$, we get that $\min \{a,b\} \le \lfloor (k-1)/2 \rfloor$.
    We select the split vertex $w$ to be the endpoint of $\hat e$ on
    the side that has fewer, i.e., $\min\{a,b\}$, crossings along~$\hat e$.
    We now assume $w_j < w < v$. The case $u < w < w_j$ is symmetric.

    We bound the piercing number of the link $\{v, w\}$.  To this end,
    we start tracing a curve $c''$ at $v$ and follow the
    link~$\lambda$ until just before it crosses $e_j$, then we
    follow~$e_j$ until just before it crosses $\hat e$, and then we
    follow~$\hat e$ until we arrive at~$w$; see \cref{fig:okp-1.5k-proof-1}.
        In the first part, the curve crosses at most $\lceil(\ell-1)/2\rceil \le \lceil(k-1)/2\rceil$ edges.
        The curve does not cross any edge in the middle part, as the
        crossing between $e_j$ and $\hat e$ is the first crossing
        along $e_j$ starting from~$\lambda$.
        In the last part, the curve crosses at most $\min\{a, b\} \le \lfloor (k-1)/2 \rfloor$ edges.
    Thus, by \cref{obs:straight-edges}, the piercing number of the link $\{v, w\}$ is at most $\lceil(k-1)/2\rceil + \lfloor (k-1)/2 \rfloor \le k-1$.

    To bound the piercing number of the link $\{u, w\}$, we start
    tracing a curve~$c'$ at~$u$, follow the link~$\lambda$ until just
    after it crosses $e_j$ and then follow the curve~$c''$ to~$w$, see
    \cref{fig:okp-1.5k-proof-1}.  The curve~$c'$ crosses the same set
    of edges as the curve~$c''$ plus the edge~$e_j$, so $c''$ has at
    most $k$ crossings.
    Thus, the piercing number of the link $\{u, w\}$ is at most $k$.

    Now we consider the last case, where $\hat e$ and $\lambda$ have a
    common endpoint.  In this case, we choose the other endpoint of
    $\hat e$ as $w$ (see \cref{fig:okp-1.5k-proof-2}) and analyse the
    two subcases.

    We first assume that $u$ is the common endpoint.  Then the link
    $\{u, w\}=\hat{e}$ is an edge of~$G$, so it has at most $k$ crossings.
    To bound the piercing number of the link $\{v, w\}$, we argue with
    the curve~$c''$ that goes from~$v$ to~$w$ as defined in the
    previous case.  Observe that
    the edge $\{u, w\}$ has exactly $j$ crossings on the part
    from~$u$ to~$x$ (including~$x$).  Therefore, we can argue as before
    and conclude that the piercing number of the edge $\{v, w\}$ is
    at most $(\ell - j) + (k - j) = k - (2j - \ell)$.
    This is bounded from above by $k-1$
    since $2j - \ell = 2 \lceil(\ell+1)/2\rceil - \ell$ is $1$
    if $\ell$ is odd and $2$ otherwise.
    
    Similarly, if $v$ is the common endpoint, then the link
    $\{v,w\}=\hat{e}$ is an edge of~$G$, so it has at most $k$ edges,
    and the piercing number of the link $\{u, w\}$ is at most
    $(j-1) + (k - (\ell - j + 1)) = k + (2j - \ell) - 2 \le k$.
\end{proof}

The \emph{piercing number} of a triangle of~\T is the sum of the
piercing numbers of the links that form the sides of the triangle.
We define the \emph{triangle piercing number} of~\T to be the maximum
piercing number of any triangle of~\T.
Note that the triangle piercing number of a triangulation is at most
three times its edge piercing number, but it can be smaller.
If $k$ is odd, we obtain a slightly stronger result;
if $k$ is odd, for every triangle $t$ that we create in
the proof of \cref{lem:triangulation-strong}, at least one edge of~$t$
has at most $k-1$ crossings.

\begin{lemma}\label{lem:triangulation-strong-odd}
  For every odd $k \ge 1$, 
  every outer $k$-planar drawing of a maximal outer $k$-planar graph
  admits a triangulation of the outer cycle with
  edge piercing number at most $k$ 
  and triangle piercing number at most $3k-1$.
\end{lemma}

\begin{proof}
Following the proof of \cref{lem:triangulation-strong}, in
all cases, there is an edge with piercing number at most
$k-1$. If $\lambda$ has no piercing edges, its piercing
number is $0 \le k-1$. Otherwise, if $\hat e$ and $\{u,v\}$
are disjoint, $\{v,w\}$ is pierced by at most $k-1$ edges.
In the last case, where $\hat e$ and $\{u,v\}$
have a common endpoint, $\{v, w\}$ is pierced by at most
$k-1$ edges if $\ell$ is odd. If $\ell$ is even, it implies
$\ell \neq k$ and therefore the link $\lambda$ is pierced
by $\ell \leq k-1$ edges.
\end{proof}

With a slight modification of the strategy described in the proof of
\cref{lem:triangulation-strong} and a careful analysis, we also obtain
a better bound of the triangle piercing number for the case $k = 2$.

\begin{lemma}\label{lem:triangulation-o-2-p-graphs}
  Every outer $2$-planar drawing of a maximal outer $2$-planar graph
  admits a triangulation of the outer cycle with
  edge piercing number at most $2$ 
  and triangle piercing number at most $4$.
\end{lemma}

  \begin{figure}
    \centering
    \includegraphics[page=3]{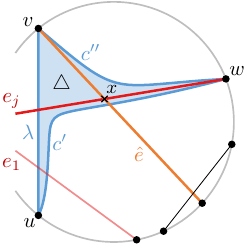}
    \caption{$\lambda$ and $\hat e$ share the endpoint $v$;
            the new triangle~$\triangle$ is bounded by~$\lambda$,
            $c'$, and~$c''$.}
    \label{fig:o2p-tpn4}
  \end{figure}

\begin{proof}
  We roughly follow the strategy of the proof of
  \cref{lem:triangulation-strong}.  Recall that $\lambda$ is the
  active link and that $\ell$ is the number of edges that
  pierce~$\lambda$.  If $\ell>0$, then~$e_j$ with
  $j=\lceil(\ell+1)/2\rceil$ is the ``middle'' edge
  among those that pierce~$\lambda$.  If
  $e_j$ has crossings on the right side of $\lambda$, then, among
  these, $x$ is the crossing point of $e_j$ closest to~$\lambda$.  The
  edge $\hat{e}$ is the edge that crosses $e_j$ in $x$.  As in the
  proof of \cref{lem:triangulation-strong}, we distinguish the
  following cases.
  \begin{enumerate}
  \item \label{enum:no} $\lambda$ has no piercing edges.
  \item \label{enum:right} $e_j$ has no crossings to the right of
    $\lambda$.
  \item \label{enum:disjoint} $\lambda$ and $\hat{e}$ are disjoint
    (\cref{fig:okp-1.5k-proof-1}).
  \item \label{enum:u} $\lambda$ and $\hat{e}$ meet in $u$
    (\cref{fig:okp-1.5k-proof-2}).    
  \item \label{enum:v} $\lambda$ and $\hat{e}$ meet in $v$.
  \end{enumerate}
  In cases \ref{enum:no}--\ref{enum:u}, we choose the split vertex~$w$
  as in the proof of \cref{lem:triangulation-strong}, which means that
  the new active links are pierced at most $k=2$ times.

  In case~\ref{enum:v}, we now choose $w$ to be $w_j$ (see
  \cref{fig:o2p-tpn4}).  Let $c'$ be a curve that starts in~$u$,
  follows $\lambda$ until just before~$e_j$ and then follows~$e_j$
  to~$w$.  Let $c''$ be a curve that starts in~$v$,
  follows $\hat e$ to~$x$,
  and then follows~$e_j$ to~$w$.  In order to show that each of the
  two curves is crossed by at most two edges, we need to do some
  preparations.
  We say that an edge lies to the \emph{left} of the active link if
  both of its endpoints lie on the left side or one endpoint lies on
  the left and the other is an endpoint of the active link.
  We further say that an edge~$e$ piercing a link $\lambda$ is
  \emph{anchored} with respect to~$\lambda$
  if there exists another edge $f$ that crosses~$e$ and either
  $f=\lambda$ or $f$ lies to the left of~$\lambda$.
  Observe that an anchored edge can have at most one crossing to the
  right of~$\lambda$.

  We claim that, during the splitting procedure described above, every
  edge piercing an active link is anchored (with respect to the active
  link).  The claim is true in the
  beginning as $\lambda=\{v_1,v_n\}$ is not pierced.  In
  case~\ref{enum:no}, every edge that pierces $\{u,w\}$ or $\{v,w\}$
  also crosses the edge $\{u,w\}$, so it is anchored.  In
  case~\ref{enum:right}, every edge that pierces $\{u,w\}$ or
  $\{v,w\}$ also pierces $\lambda=\{u,v\}$
  and is anchored by induction.  In
  case~\ref{enum:disjoint}, observe that no edge crosses $e_j$ between
  $\lambda$ and $x$ (as $x$ is the first crossing on $e_j$ to the
  right of $\lambda$), and no edge crosses $\hat e$ between $x$ and
  $w$ (by the choice of $w$ and since $\lfloor (k-1)/2 \rfloor = 0$
  for $k=2$).  Therefore, every edge that pierces $\{u,w\}$ or
  $\{v,w\}$ also pierces $\lambda$ and is anchored by induction.
  In case~\ref{enum:u}, $\{u,w\}$ is an edge, so every 
  edge piercing it is anchored.  Every edge that pierces $\{v,w\}$
  either crosses $\{u,w\}$ or pierces~$\lambda$ and is anchored by
  induction.  Similarly, in case~\ref{enum:v},
  every edge that pierces $\{u,w\}$ or $\{v,w\}$ either
  pierces~$\lambda$ or crosses~$e_j$ (which lies to the left of both
  new active links).

  We now return to case~\ref{enum:v} and bound the piercing numbers of
  the new active links $\{u,w\}$ and $\{v,w\}$.  To do this, we count
  the edges that cross the curves~$c'$ and $c''$, respectively.
  As~$e_j$ is anchored with respect to $\lambda$, we get that $x$ is
  the only crossing on $e_j$ to the right of $\lambda$.  Now, it is
  easy to see that $c'$ is pierced only by $\hat e$ and possibly by
  $e_1$ in case $j=2$.  For $c''$, we can show that there are no
  crossings along the curve.  First, observe that any edge that
  crosses~$\hat e$ between $v$ and $x$ has to pierce~$\lambda$,
  which is not possible.  Second, no edge crosses $e_j$ between $x$
  and $w$, as $e_j$ is anchored and $x$ is the only crossing on $e_j$
  to the right of~$\lambda$.  Combining these two observations, we see
  that there is no place for $c''$ to cross with an edge.

  It remains to show that the triangle piercing number of the new
  triangle $\triangle$ formed by $u$, $v$, and $w$ is at most~$4$.
  In case~\ref{enum:no}, $\lambda$ has no piercing edges.
  In cases~\ref{enum:right}, \ref{enum:u}, and~\ref{enum:v}, $\{v,w\}$
  has no piercing edges.
  In case~\ref{enum:disjoint}, we distinguish two subcases depending on
  the relative position of $w$.
  When $w_j < w$, then $\{v,w\}$ has no piercing edges.
  When $w < w_j$, then each of $\{u,w\}$ and $\{v,w\}$ is pierced
  at most once.
\end{proof}
  
Lastly, we combine the two triangulation strategies from the proofs of
\cref{lem:triangulation-weak,lem:triangulation-strong} to obtain a
triangulation method for outer min-$k$-planar graphs.

\begin{lemma}\label{lem:triangulation-min}
  For every $k \ge 1$,
  every outer min-$k$-planar drawing of a maximal outer min-$k$-planar graph admits a triangulation of the outer cycle with 
  edge piercing number at most $2k-1$ 
  and triangle piercing number at most $6k-3$.
\end{lemma}

\begin{proof}
    We proceed as in the proof of \cref{lem:triangulation-weak}.
    We describe a splitting procedure such that every active link has
    piercing number at most $2k-1$.  Consult \cref{fig:okp-proof} for
    a reminder of the notation and the strategy.
    As the given drawing is outer min-$k$-planar, edges may now have
    more than $k$ crossings.
    We call such edges \emph{heavy} and the other edges \emph{light}.
    By definition, heavy edges cross only light edges.

    We first explain why the previous proof breaks in both cases if
    the curves $c'$ and $c''$ are routed along heavy edges.  Recall
    that in the first case when $\lambda$ has no piercing edges, we
    selected the split vertex $w$ to be the neighbor of $u$ from the
    right side of $\lambda$ that is closest to~$v$.  If, however, the
    edge $\{u,w\}$ is heavy, then we cannot bound the number of
    crossings along the curve~$c''$.  In the second case, we selected
    the split vertex to be the right endpoint~$w_j$ of the middle
    piercing edge~$e_j$.  If the edge~$e_j$ is heavy, then we cannot
    bound the numbers of crossings along~$c'$ and~$c''$.

    To resolve the issue with the second case, we apply the strategy
    from the proof of \cref{lem:triangulation-strong}.
    Observe that if $e_j$ is heavy, then for the first crossing $x$
    with an edge~$\hat e$ on the right side of~$\lambda$, the
    edge~$\hat e$ is a light.
    Similarly, as in the proof of \cref{lem:triangulation-strong} (see
    \cref{fig:okp-1.5k-proof}), we select the split vertex~$w$ to be
    the endpoint of~$\hat e$ that is different from~$u$ and~$v$ and in
    case both endpoints are, the one on the side of~$x$ with fewer
    crossings.
    The piercing numbers of the two new links are at most
    $\lceil(2k-1)/2\rceil+k-1 \le 2k-1$.

    To resolve the issue with the first case, we use a similar
    strategy.  If the edge $\{u,w\}$ is heavy, let $x$ be the crossing
    on $\{u,w\}$ that is closest to~$u$ and let $\hat e$ be the
    crossing edge.  Note that $\hat e$ is light.  We select the split
    vertex $w$ to be any right endpoint of $\hat e$.  Now consider the
    two curves that start at~$w$, go along~$\hat e$ to~$x$, then along
    $\{u,w\}$ to~$u$, and optionally to $v$ along~$\lambda$.  Both
    curves have at most $k$ crossings, only with edges that also
    cross~$\hat e$.
\end{proof}

\begin{remark}
  Given the intersection graph of the edges of~$G$, the triangulations
  in Lemmas~\ref{lem:triangulation-weak}--\ref{lem:triangulation-min}
  can be constructed in $O(nk)$ time.
  Each step of the splitting procedures presented in
  Lemmas~\ref{lem:triangulation-weak}--\ref{lem:triangulation-min} can
  be implemented to run in $O(k)$ time, as the edges that pierce the
  new active links also pierce at least one of the edges~$\lambda$,
  $e_j$, and $\hat e$.
\end{remark}

\section{Applications}
\label{sec:applications}

In this section, we show several applications of our triangulations 
from Section~\ref{sec:triangulations}. Among others, 
we improve the bound of Wood and Telle~\cite{wt-pdcng-NYJM07} on 
the treewidth of outer $k$-planar graphs from
$3k+11$ to $1.5k + 2$ (see~\cref{sec:treewidth}) and
improve the bound of Chaplick et al.~\cite{BeyondOuterplanarity} 
on the separation number of outer $k$-planar graphs from $2k+3$ 
to $k + 2$; see \cref{sec:separation}. 
In addition,  we improve 
the bound of Wood and Telle of $3k+11$ on the treewidth of outer min-$k$-planar graphs to $3k+1$; see \cref{thm:min-okp-tw3k}.

\subsection{Treewidth}
\label{sec:treewidth}

Now we show our main tool to obtain a better upper bound
on the treewidth. Note that we do not refer specifically
to outer $k$-planar graphs since our tool can be
used for any drawing of a graph whose vertices lie on the outer
face. 

Consider a graph $G$ with a drawing $\Gamma$ whose outer face is
bounded by a simple cycle~$C$ that contains all vertices of~$G$.
Let~$\T$ be a planar triangulation of~$C$ and $f$ be some triangular face of $\T$. 
We call an edge of $E(G)\setminus E(\T)$ \emph{short} with respect to
$f$ if one of its endpoints is on $f$
and it pierces a link of $f$. We call an edge of $E(G)\setminus E(\T)$ \emph{long} with respect to $f$ 
if neither of its endpoints is on $f$ and it pierces two links
of $f$.

\begin{lemma}\label{lem:triangulation-treewidth}
  If $G$ is a graph with a convex drawing whose outer cycle admits
  a triangulation~\T with triangle piercing number at
  most $c$, then $\tw(G) \le (c + 5) / 2$.
\end{lemma}

\begin{proof}
  Let $T$ be a tree
  obtained by taking the weak dual graph of~\T and associating each node of $T$ with a bag that contains the vertices of $G$ that are
  incident to the corresponding face.  We modify this tree to consider %
  the piercing edges, ensuring that the sizes of bags do not exceed $(c+7)/2$. 
  (Recall that the
  width of a tree decomposition is the size of the largest bag
  minus~1.)

  If \T has an unpierced inner link~$\lambda=\{u,v\}$, we split~\T along~$\lambda$
  into two triangulations and compute tree decompositions for them.
  Since each of them has a bag that contains $\lambda$, we can
  connect them and obtain a tree decomposition of $G$. 

  Hence, from now on, we assume that every inner link of~\T is pierced at least once.
  We root~$T$ at an arbitrary leaf.  The basic idea is as follows.  Consider an edge~$\{u,v\}$
  that pierces inner links of~\T.  Then we \emph{lift} both~$u$
  and~$v$ to the bag~$b$ that is the lowest common ancestor of the highest bags~$b_u$
  and~$b_v$ that contain $u$ and $v$, respectively. That is, we place a copy of~$u$ into each bag on the way
  from~$b_u$ to~$b$ and symmetrically for~$v$.
  If we apply this strategy naively, then there may be a bag with
  $c + 3$ vertices in the worst case.  
  In the following, we describe how to improve this bound. 
  
  Consider a face $f$ of the triangulation $\T$. Let $b_f$ be the bag
  in~$T$ that corresponds to~$f$.  Let 
  $v_1$, $v_2$, and $v_3$ be the vertices of~$f$ such that $f$ shares the link $\{v_1,v_2\}$ 
  with the face that corresponds to the parent of~$f$ in~$T$ (if any). 
  An edge of $E(G)\setminus E(\T)$ is \emph{lineal} 
  if one of its endpoints is in
  the bag $b_f$ or above and the other endpoint is in 
  a bag below~$b_f$ in the tree~$T$.
  A long edge of $E(G)\setminus E(\T)$ is \emph{bent} if its endpoints are in 
  different subtrees of~$T$. 
  
  We perform the following two modifications of~$T$ for each face~$f$; see \cref{fig:step-1-2}.

  \begin{itemize}
  \item We subdivide each edge of~$T$ between $b_f$ and its children (if any) by a new
    \emph{copy bag} that is a copy of~$b_f$.  We call these bags
    \emph{primary} copy bags.
  \item\label{item:tree-decomp-step2} 
    If $f$ has two children in~$T$, assume that piercing number of $\lambda=\{v_1,v_3\}$ 
    is less or equal to that of~$\lambda'=\{v_2,v_3\}$ and that $\lambda$ is the
    common link of~$f$ and its left child face.  If there are $l>0$
    bent edges $\{u_1,w_1\},\dots,\{u_l,w_l\}$ w.r.t. $f$, then we subdivide the edge of~$T$ between~$b_f$
    and its left child by additional $l$ copies
    of~$b_f$.  These \emph{secondary} copy bags are numbered $b_{f,1}, \dots, b_{f,l}$ from top
    to bottom; see \cref{fig:step-1-2}.
  \end{itemize}

  \begin{figure}[tb]
    \centering
    \includegraphics{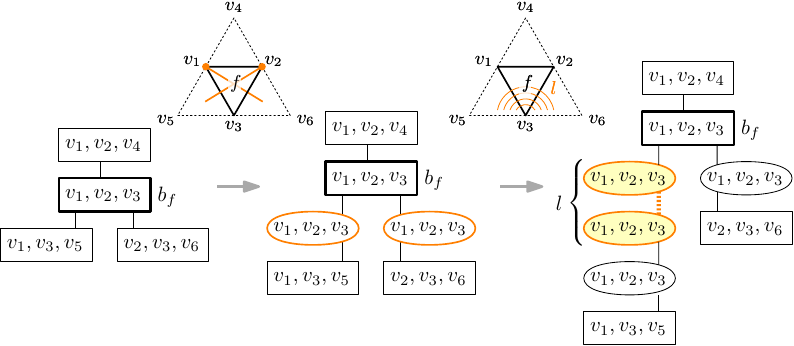}
    \caption{Steps 1 and 2 of the modification.  Copy bags are oval;
      secondary copy bags are yellow.}
    \label{fig:step-1-2}
  \end{figure}

  Next, we lift vertices in bottom-up order, e.g., in preorder.  More
  precisely, at each \emph{original} bag $b_f$ (that is, not a copy
  bag) in our preorder traversal
  of~$T$, we process each edge~$e$ of~$G$ that pierces~$f$, according
  to its type as follows.

  \begin{itemize}
  \item If~$e$ is a short edge, that is, it is incident to a vertex
    of~$f$, namely to $v_1$ or $v_2$, then we lift the highest
    occurrence of the other endpoint~$u$ of~$e$ by one bag (as $v_5$
    and $v_7$ in \cref{fig:td-add-crossing-edge-2})
    to the primary copy bag of~$b_f$. %
  \item If $e$ is a long edge that pierces the link~$\{v_1,v_2\}$, then we
    lift the highest occurrence of the other endpoint~$u$ of~$e$ all
    the way to~$b_f$; see \cref{fig:td-add-crossing-edge-3}.
  \item If $e=\{u_i,w_i\}$ for some $i \in \{1,\dots,l\}$, then we
    lift both endpoints of~$e$.  First, we lift the highest occurrence
    of~$u_i$ from the left subtree to~$b_{f,i}$.  Second, we lift the
    highest occurrence of~$w_i$ in the right subtree of~$b_f$
    to~$b_f$.  Third, we copy~$w_i$ into each bag from~$b_f$ down
    to~$b_{f,i}$; see \cref{fig:td-add-crossing-edge-5}.  In this way,
    the two endpoints of~$e$ meet in~$b_{f,i}$.
  \end{itemize}
  
  \begin{figure}[tb]
    \begin{subfigure}[t]{.273\textwidth}
      \centering
      \includegraphics[page=1]{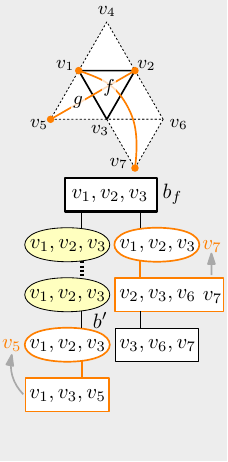}
    \end{subfigure}
    \begin{subfigure}[t]{.273\textwidth}
      \centering
      \includegraphics[page=2]{lifting-up}
    \end{subfigure}
      \begin{subfigure}[t]{.42\textwidth}
      \centering
      \includegraphics{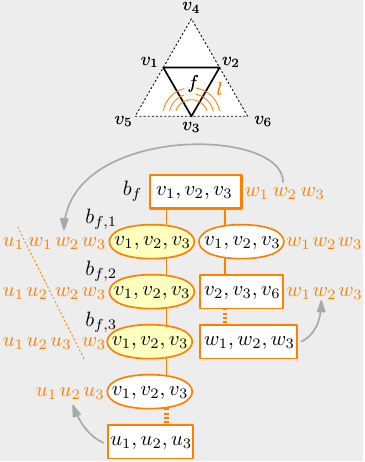}
    \end{subfigure}

    \vspace*{-2.5ex}
    
    \begin{subfigure}[t]{.275\textwidth}
      \subcaption{}
      \label{fig:td-add-crossing-edge-2}
    \end{subfigure}
    \begin{subfigure}[t]{.275\textwidth}
      \subcaption{}
      \label{fig:td-add-crossing-edge-3}
    \end{subfigure}
    \begin{subfigure}[t]{.42\textwidth}
      \subcaption{}
      \label{fig:td-add-crossing-edge-5}
    \end{subfigure}

    \caption{Lifting up endpoints of the edges that pierce links of
      the face~$f$.}
    \label{fig:lifting-up}
  \end{figure}
  
  The modifications described above make sure that the resulting tree $T'$ 
  is indeed %
  a tree decomposition of~$G$, i.e.,
  for each vertex~$v$ of~$G$, the subtree induced by the bags
  that contain~$v$ is connected, and for each edge~$e$ of~$G$, there
  is a bag in~$T'$ that contains both endpoints of~$e$.

In the following, we show that each bag of $T'$ contains 
at most $(c+ 7)/ 2$  vertices, i.e., we have added 
at most $(c+1) / 2$ to the initial three vertices.
For each face $f$, the number of vertices added to the
original bag~$b_f$ is the number of long edges w.r.t.~$f$, which is at
most $c/2$.  It remains to bound the number of vertices that are added
to each copy bag of~$b_f$.

First, we show the bound for secondary copy bags.  Recall that, for
$i \in \{1, \dots, l\}$, we have added to~$b_{f,i}$ from below $i$
endpoints of the bent long edges w.r.t.~$f$ and from above $l+1-i$
endpoints of bent long edges; see~\cref{fig:td-add-crossing-edge-5}.
In addition, we have added a vertex for each lineal long edge.
Therefore, the number of added vertices is at most one plus the number
of long edges w.r.t.~$f$ that pierce the link $\{v_1,v_3\}$.  Recall
that the link $\{v_1,v_3\}$ was chosen to have at most as many
piercing edges as $\{v_2,v_3\}$.  Because each inner link of~\T, including
$\{v_1,v_2\}$, is pierced at least once, $\{v_1,v_3\}$ is pierced at most
$(c-1)/2$ times.  Thus, we have added at most $(c+1)/2$ vertices
to $b_{f,i}$, and so $|b_{f,i}| \le (c+7)/2$.

Now, we bound the number of vertices that are added to the primary (the
bottommost) copy bag $b'$ of $b_f$. Note that such a bag is a parent
bag of some original bag $b_g$ such that the faces $f$ and $g$ are
adjacent.  There are two types of vertices that are added to~$b'$.
The first type consists of the endpoints of the lineal long edges of $b_g$ (at most $c/2$).
The second one is just one vertex $w$ of the face $g$ that is not shared with $f$ (for example, the endpoint $v_5$ of $g$ in~\cref{fig:td-add-crossing-edge-2}). 
We have added~$w$ only if there are edges in $E(G)\setminus E(\T)$ with lower endpoint $w$ (lower w.r.t.~$T$).
Each such edge contributes to the piercing number of~$g$. Hence, the number of lineal long edges w.r.t.~$g$ is at most $(c-1)/2$.
In both cases ($w \in b'$ and $w \not\in b'$), we have added at most $(c+1)/2$ vertices to~$b'$, so $|b'| \le (c+7)/2$.
\end{proof}

\cref{lem:triangulation-o-2-p-graphs,lem:triangulation-treewidth} give us an upper bound $4$ for outer $2$-planar graphs,
extending the known tight bound $k + 2$ for $k = \{0,1\}$
\cite{auer2016outer,abbghnr-co1pg-Algorithmica21}.
Note that $K_5$ is outer $2$-planar and $\tw(K_5) = 4$.

\begin{theorem}
  \label{thm:o2p-tw4}
  Every outer 2-planar graph has treewidth at most $4$,
  which is tight.
\end{theorem}

\begin{theorem}
	\label{thm:okp-tw1.5k}
	Every outer $k$-planar graph has treewidth at most $1.5k + 2$.
\end{theorem}
\begin{proof}
    We already know better bounds for $k \le 2$.
    For $k \ge 3$, by \cref{lem:triangulation-strong,lem:triangulation-treewidth,lem:triangulation-strong-odd}, we obtain a bound $1.5k + 2.5$ for even $k$ and $1.5k + 2$ for odd $k$.
    As treewidth is an integer, we obtain $1.5k + 2$ for general $k$.
\end{proof}

By \cref{lem:triangulation-min,lem:triangulation-treewidth}, we also obtain the following bound for outer min-$k$-planar graphs, which improves the previously known bound $3k+11$ by Wood and Telle
\cite{wt-pdcng-NYJM07}.
Note that outer min-$0$-planar graphs are outerplanar graphs.
\begin{theorem}
  \label{thm:min-okp-tw3k}
  For $k \ge 1$, every outer min-$k$-planar graph has treewidth at
  most $3k+1$.  %
\end{theorem}

\subsection{Separation Number}
\label{sec:separation}

\cref{thm:okp-tw1.5k,thm:min-okp-tw3k} immediately imply upper bounds
on the separation number of outer $k$-planar and outer min-$k$-planar
graphs.  However, using our triangulations directly, we obtain even
better bounds, namely $k + 2$ for outer $k$-planar graphs and $2k + 1$
for outer min-$k$-planar graphs.  The first bound improves the bound
of $2k+3$ by Chaplick et al.~\cite{BeyondOuterplanarity}.

\begin{lemma}{\label{lem:triangulation-balanced-separator}}
  If $G$ is a graph with a convex drawing whose outer cycle admits
  a triangulation~\T with edge piercing number at most $c$,
  then $G$ has a balanced separator of size at most $c + 2$.
\end{lemma}

\begin{proof}
  We construct a balanced separation of order at most $c + 2$ using~\T.
  In short, we select a link~$\lambda$ of~\T that is a balanced separator for~\T and
  put its endpoints into the separator.  Then, we add at most $c$
  vertices to the separator according to the edges piercing~$\lambda$.

  First, we find a ``centroid'' triangle of~\T as follows.  Let $T$ be
  the tree that is the weak dual of~\T.  It is well known that every
  tree contains a vertex such that, after removing it, the number of
  vertices in each subtree is at most half of the original tree.  Let
  $f = \{u, v, w\}$ be the triangle corresponding to such a vertex
  of~$T$.
We partition $V \setminus f$ into three disjoint sets $V_1, V_2, V_3$; see \cref{fig:separation-number-1}.
We may assume that $V_1$ is the largest among the three sets.

\begin{figure}[h]
    \begin{subfigure}[b]{0.48\textwidth}
      \centering
      \includegraphics[page=1]{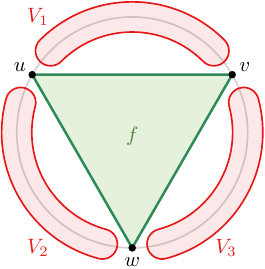}
      \subcaption{``centroid'' triangle $f$ and vertex sets $V_1$,
        $V_2$, $V_3$.}
        \label{fig:separation-number-1}
    \end{subfigure}
    \hfill
    \begin{subfigure}[b]{0.48\textwidth}
      \centering
      \includegraphics[page=2]{separation-number}
      \subcaption{\nolinenumbers{}edges piercing $\{u, v\}$ and
        separation~$(A, B)$.}
      \label{fig:separation-number-2}
    \end{subfigure}
    \caption{A triangulation with edge piercing number~$c$ yields a
      balanced separator of size at most~$c+2$.}
    \label{fig:separation-number}
\end{figure}

Now, we construct the desired separation with the help of the sets
$V_1, V_2, V_3$.  As $\lambda=\{u, v\}$ is a link of~\T, at most $c$ edges
pierce~$\lambda$, connecting vertices in~$V_1$ with vertices in
$V_2 \cup \{w\} \cup V_3$.  Let $S$ be the set of endpoints of the
piercing edges on the latter side.
Let $A = V_1 \cup \{u, v\} \cup S$ and $B = V_2 \cup
\{u,v,w\} \cup V_3$; see \cref{fig:separation-number-2}.
We claim that $(A, B)$ is a balanced separation of order at most $c + 2$.

Clearly, the order of $(A, B)$ is $|A \cap B| = |S \cup \{u, v\}|$, which is at most $c + 2$.
Hence, it suffices to show that the sizes of $A \setminus B$ and $B \setminus A$ are at most $2n/3$.
To this end, we first bound the size of $V_1$ as follows:
\begin{align*}
\frac{n}{3} - 1 \leq |V_1| \leq \frac{n}{2}-1.
\end{align*}
The lower bound holds as $|V_1| + |V_2| + |V_3| = n - 3$ and $|V_1|$
is the largest.
The upper bound can be shown by the fact that a triangulated outerplanar graph has exactly $F + 2$ vertices, where $F$ is the number of triangles.
By the way of choosing $f$, the vertex set $V_1 \cup \{u, v\}$ induces at most $n/2-1$ triangles.
Therefore, $|V_1| = | V_1 \cup \{u, v\}| - 2 \leq (n/2-1 + 2) - 2 = n/2-1$.

Now we can confirm that the sizes of $A \setminus B$ and
$B \setminus A$ are at most $2n/3$ as follows:
\begin{align*}
|A \setminus B| &= |V_1| \leq \frac{n}{2} - 1 \leq \frac{2}{3} n
                  \text{~ and} \\
|B \setminus A| &\leq |V_2 \cup \{w\} \cup V_3| = |V \setminus (V_1 \cup \{u, v\})| \leq n - \left(\frac{n}{3} - 1 \right) - 2 = \frac{2}{3} n - 1. \qedhere
\end{align*}
\end{proof}

For every $k$, the classes of outer $k$-planar graphs and outer min-$k$-planar graphs are closed under taking subgraphs.
Hence, by \cref{lem:triangulation-balanced-separator,lem:triangulation-strong,lem:triangulation-min},
we obtain the following bounds.

\begin{theorem}
  \label{thm:separation-number-okp}
  \label{thm:separation-number-o-mink-p}\sloppy
  Every outer $k$-planar graph has separation number at most $k + 2$,
  and for $k \geq 1$, every~outer min-$k$-planar graph has separation
  number at most $2k+1$.%
\end{theorem}

\section{Lower Bounds}
\label{sec:lower}

In this section, we complement the results in \cref{sec:applications}, by giving lower bounds for the separation number and treewidth of outer $k$-planar graphs.
To show them, we use grid-like graphs called \emph{stacked prisms}.
The $m \times n$ stacked prism $Y_{m,n}$ is the (planar) graph obtained
by connecting all pairs vertices in the topmost and bottommost row of the $m \times n$ grid that are in the same column; see \cref{fig:stacked-prism-grid-like}.

\begin{figure}
  \begin{subfigure}[b]{.47\textwidth}
    \centering
    \includegraphics[page=1]{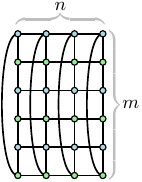}
    \subcaption{grid-like drawing}
    \label{fig:stacked-prism-grid-like}
  \end{subfigure}
  \hfill
  \begin{subfigure}[b]{.47\textwidth}
    \centering
    \includegraphics[page=2]{stacked-prism}
    \subcaption{\nolinenumbers{}clipping of a circular drawing}
    \label{fig:stacked-prism-circular}
  \end{subfigure}
  \caption{Two drawings of the stacked prism $Y_{m,n}$.}
  \label{fig:stacked-prism}
\end{figure}

\begin{theorem}
  \label{thm:sn-lb-okp}
  For every even number $k \geq 0$, there exists a graph $G$ such that
  $\lcro(G) = k$ and $\sn(G) = k + 2$.
\end{theorem}

\begin{proof}
  We first observe that, if $m$ is even, $\lcro(Y_{m, n})$ is at most $2n-2$.
  If $m$ is even, we can place the $m$ rows of $Y_{m,n}$, each of length~$n$, one after the other, alternating in their direction, around a circle; see
  \cref{fig:stacked-prism-circular}.
  The grid has two types of edges: the \emph{row edges} and the \emph{column edges}, which connect vertices in the same row or column, respectively.
  As \cref{fig:stacked-prism-circular} shows, the row edges have no crossing.
  For each $i \in [n]$, there is a column edge~$e_i$ that spans $2i-2$ other vertices along the perimeter of the circle. 
  Each of these vertices is incident to exactly one column edge that crosses~$e_i$.
  Hence, every column edge has at most $2n-2$ crossings, and $\lcro(Y_{m,n}) \le 2n-2$.

  Let $n=k/2 + 1$, and let $m$ be a sufficiently large even number.
  Then we claim that $Y_{m, n}$ fulfills the conditions.
  As we discussed $\lcro(Y_{m, n}) \leq k$ holds and therefore $\sn(Y_{m, n}) \leq k + 2$ follows from \cref{thm:separation-number-okp}.
  Hence, it suffices to show $\sn(Y_{m, n}) \geq k + 2 = 2n$, which also implies $\lcro(Y_{m,n}) \geq 2n-2 = k$ by \cref{thm:separation-number-okp}. 

  Suppose that $Y_{m, n}$ has a balanced separator $S$ of size less than $2n$.
  As there are $n$ columns, there is a column that contains at most one vertex of $S$.
  This column contains a path $P$ of length $m-1$ that does not intersect $S$.
  Now observe that at least $m - 2n$ rows do not contain any vertex of $S$, and therefore, all vertices in these rows are connected to $P$.
  Hence, after removing $S$, the size of the connected component $C$ that contains $P$ is at least $n (m-2n)$.
  The ratio $|V(C)| / |V(Y_{m, n})|$ is $1 - (2n / m)$, which is greater than $2/3$ if $m$ is sufficiently large.
\end{proof}

Aidun et al.~\cite{ToroidalGridTW} showed that $\tw(Y_{m,n})=2n$ if $m > 2n$.
With the same stacked prism, we obtain the following lower bound.

\begin{theorem}
  \label{thm:tw-lb-okp}
  For every even number $k \geq 0$, there exists a graph $G$ such that $\lcro(G) = k$ and $\tw(G) = k + 2$.
\end{theorem}

\section{Conclusion and Open Problems}

We have introduced methods for triangulating drawings of outer
$k$-planar graphs such that the triangulation edges cross few graph
edges.  These triangulations yield better bounds on
treewidth and separation number of outer $k$-planar graphs.  Our
method is constructive; the corresponding treewidth decomposition and
balanced separation can be computed efficiently.
Via our triangulations, we improved the multiplicative constant in the
upper bound on the treewidth of outer $k$-planar graphs from~$3$
to~$1.5$; we showed a lower bound of~$1$.  What is the correct
multiplicative constant?
Finding triangulations of outer $k$-planar graphs with lower triangle
piercing number could be a step in this direction.
It would also be interesting to find other graph classes that admit
triangulations with low triangle piercing number.

\bibliography{main}

\end{document}